\documentclass{ecai} 

\usepackage{amssymb}
\usepackage{amsmath}
\usepackage{amsthm}

\usepackage{algorithm}
\usepackage{thm-restate}
\usepackage{mathtools}
\usepackage[noend]{algpseudocode}
\usepackage{amsfonts}
\usepackage{tikz}
\usepackage{subfigure}

\usepackage[english]{babel}
\usepackage[autostyle, english = american]{csquotes}
\MakeOuterQuote{"}

\tikzset{blkvertex/.style={black, draw=black, circle, fill, scale=0.5}}

\tikzset{sqvertex/.style={black, draw=black, rectangle}}

\tikzset{circlevertex/.style={black, draw=black, circle, fill, scale=0.5}}

\tikzset{emptyvertex/.style={black, draw=black, circle, scale=0.5}}

\tikzset{diamondvertex/.style={black, draw=black, diamond, fill, scale=0.5}}

\usetikzlibrary{decorations.markings}
\usetikzlibrary{arrows.meta}
\tikzset{arc/.style={decoration={
			markings,
			mark=at position 0.6 with {\arrow{>[scale=2.5,
					length=2,
					width=2.5]}}},postaction={decorate}}}

\newtheorem{theorem}{Theorem}
\newtheorem{lemma}[theorem]{Lemma}

\newtheorem{proposition}[theorem]{Proposition}

\newtheorem{definition}{Definition}

\newcommand{\BibTeX}{B\kern-.05em{\sc i\kern-.025em b}\kern-.08em\TeX}

\begin{document}


\begin{frontmatter}


\paperid{0561} 


\title{Polynomial-Time Algorithms for Fair Orientations of Chores}


\author[A]{\fnms{Kevin}~\snm{Hsu}\orcid{0000-0002-8932-978X}\thanks{Corresponding Author. Email: kevinhsu996@gmail.com}}
\author[A]{\fnms{Valerie}~\snm{King}\orcid{0000-0001-7311-7427}}

\address[A]{University of Victoria}


\begin{abstract}
This paper addresses the problem of finding fair orientations of graphs of chores, in which each vertex corresponds to an agent, each edge corresponds to a chore, and a chore has zero marginal utility to an agent if its corresponding edge is not incident to the vertex corresponding to the agent. Recently, Zhou~et~al.\ (IJCAI,~2024) analyzed the complexity of deciding whether graphs containing a mixture of goods and chores have EFX orientations, and conjectured that deciding whether graphs containing only chores have EFX orientations is NP-complete. We resolve this conjecture by giving polynomial-time algorithms that find EF1 and EFX orientations of graphs containing only chores if they exist, even if there are self-loops. Remarkably, our result demonstrates a surprising separation between the case of goods and the case of chores, because deciding whether graphs containing only goods have EFX orientations was shown to be NP-complete by Christodoulou et al.\ (EC,~2023). In addition, we show the EF1 and EFX orientation problems for multigraphs to be NP-complete.
\end{abstract}

\end{frontmatter}


\section{Introduction}

Consider the problem of dividing a set of items among a group of agents with the stipulation that each item must be given to exactly one agent and cannot be shared. This problem models a wide variety of settings, such as assigning offices in a building to employees, allocating classrooms to multiple simultaneous lectures, and even divorce settlements between spouses. The classical example of dividing a single good between two agents shows that it is not always possible to find an envy-free allocation, necessitating the search for alternative criteria of fairness.

The most prominent alternative is envy-freeness up to any good (EFX) \citep{caragiannis2019unreasonable,gourves2014near}, which allows for each agent $i$ to envy each agent $j$, as long as the envy can be alleviated if the good in $j$'s bundle that has the smallest marginal utility for $i$ is removed from $j$'s bundle. The existence problem of EFX allocations has proven to be notoriously difficult and has been referred to as "fair division's most enigmatic question" by \citet{procaccia2020technical}. Fruitful results have been obtained by restricting instances. For example, EFX allocations exist if the agents have identical utility functions or identical ranking of the goods~\citep{plaut2020almost}, if the agents have lexicographic utility functions~\citep{hosseini2023fairly}, or if there are at most 3 agents~\citep{chaudhury2020efx,akrami2024efx}.

Another possibility is to limit the number of agents who are interested in each item. The assumption that each item provides non-zero marginal utility to at most two agents leads to the graphical instances introduced by \citet{christodoulou2023fair}. These instances can be represented by a graph $G$ wherein each vertex (resp.\ edge) corresponds to an agent (resp.\ good), and assuming that whenever an edge is not incident to a vertex, it provides zero marginal utility to the agent corresponding to that vertex. In this setting, it is natural to ask whether EFX {\em orientations} exist, i.~e.\ whether it is possible to allocate each edge to one of its endpoints while satisfying the EFX condition.  Unfortunately, it is NP-complete to decide whether EFX orientations exist for instances containing only goods \citep{christodoulou2023fair}.

For any result that holds for goods, one naturally ask what happens if chores are introduced. To this end, \citet{ijcai2024p338} considered the EFX orientation problem in the mixed manna setting involving both goods and chores, and obtained a variety of complexity results for different variants of EFX in this setting. They conclude with a discussion on the special case of instances containing only chores, which can be used to model situations in which tasks are to be allocated to agents, but some tasks may not suitable to an agent because the agent is far away physically. For instance, a food delivery platform should avoid assigning orders to workers who are far away. Similarly, in a rental apartment where each floor has a set of communal chores, one should avoid assigning a chore to a tenant of a different floor.

Two variants of EFX for chores have been defined. Let $\pi$ be an allocation, $\pi_i$ denote the set of chores given to $i$, and $u_i$ denote the utility function of agent $i$. We say that $\pi$ is
\begin{itemize}
	\item EFX\_ if for each pair of agents $i \neq j$ such that $i$ envies $j$ and each $e \in \pi_i$ such that $u_i(e) < 0$, we have $u_i(\pi_i \setminus \{e\}) \geq u_i(\pi_j)$; and
	\item EFX$_0$ if for each pair of agents $i \neq j$ such that $i$ envies $j$ and each $e \in \pi_i$ such that $u_i(e) \leq 0$, we have $u_i(\pi_i \setminus \{e\}) \geq u_i(\pi_j)$.
\end{itemize}
The difference is that while EFX\_ only requires the envy be alleviated when ignoring chores with strictly negative marginal utility, EFX$_0$ requires that the envy be alleviated even if we ignore a chore with zero marginal utility. Clearly, EFX$_0$ is a strictly stronger condition than EFX\_.

\citet{ijcai2024p338} found a polynomial-time algorithm that decides whether an EFX\_ orientation exists, and conjectured that deciding whether an EFX$_0$ orientation exists is NP-complete, similar to the case involving only goods. (This conjecture appears in the arXiv version of their paper \citep{li2024complete} and not the conference version \citep{ijcai2024p338}.)

\subsection{Our Contribution}

We study the complexity of finding EF1 and EFX$_0$ orientations of chores on graphs and multigraphs that possibly contain self-loops. Importantly, we resolve the conjecture made by \citet{ijcai2024p338} by giving a polynomial-time algorithm that decides the existence of an EFX$_0$ orientation of chores and outputs one if it exists. Specifically,
\begin{enumerate}
	\item There exists a $O(|V(G)|+|E(G)|)$-time algorithms that decide whether a graph $G$ of chores has an EF1 orientation and outputs one if it exists (Theorem~\ref{thm:mainContributionThreeB}).
	\item There exists a $O((|V(G)|+|E(G)|)^2)$-time algorithms that decide whether a graph $G$ of chores has an EFX$_0$ orientation and outputs one if it exists (Theorem~\ref{thm:main}).
	\item Deciding whether a multigraph of chores has an EF1 (resp.\ EFX$_0$) orientation is NP-complete (Theorems~\ref{thm:main-multi} and \ref{thm:main-no-self-loop}).
\end{enumerate}

Remarkably, our results for chores stand in contrast with previous results for goods. While EF1 orientations of goods always exist \citep{deligkas2024ef1}, our results imply they can fail to exist for chores. Second, while deciding whether EFX$_0$ orientations of goods exist is NP-complete \citep{christodoulou2023fair}, out results imply the same problem is polynomial-time decidable for chores.

\subsection{Related Work}
While EFX allocations always exist \citep{christodoulou2023fair}, deciding whether EFX orientations exist is NP-complete \citep{christodoulou2023fair}, even for graphs with a vertex cover of size 8 \citep{deligkas2024ef1}. EFX orientability has also been shown to be related to the chromatic number $\chi(G)$ of $G$ --- while they always exist if $\chi(G) \leq 2$, they can fail to exist if $\chi(G) \geq 3$ \citep{zeng2024structure}.

The multigraph setting has also been studied. While deciding whether EFX orientations exist is NP-complete for multigraphs with only 10 vertices \citep{deligkas2024ef1}, some positive results have also been obtained. A multigraph is {\em bivalued} if there exists two values $\alpha$ and $\beta$ such that $u_i(e) \in \{\alpha, \beta\}$ for each vertex $i$ and each edge $e$. A multigraph is {\em symmetric} if for each edge $e = \{i, j\}$, we have $u_i(e) = u_j(e)$. EFX orientations exist for bivalued symmetric multigraphs not containing a structure called a non-trivial odd multitree \citep{hsu2024efx}. EFX allocations of a multigraph $G$ exist if $G$ is bipartite \citep{afshinmehr2024efx,bhaskar2024efx,sgouritsa2025existence}, a multicycle \citep{afshinmehr2024efx}, a multitree \citep{bhaskar2024efx}, if $G$ has girth $2 \chi(G) -1$ \citep{bhaskar2024efx}, if each vertex has at most $\lceil |V(G)|/4 \rceil -1$ neighbours \citep{sgouritsa2025existence}, or if the shortest cycle with non-parallel edges has length at least 6 \citep{sgouritsa2025existence}.

\subsection{Organization}

We introduce the model in Section~\ref{section:prelim}. We give our algorithms for finding EF1 and EFX$_0$ orientations of graphs in Sections~\ref{section:EF1} and \ref{section:graphs}. We prove our hardness results (Theorem~\ref{thm:main-multi} and Theorem~\ref{thm:main-no-self-loop}) in Section~\ref{section:multigraphs}. We discuss the special case of additive utility functions in Section~\ref{section:additive}. Finally, we conclude the paper in Section~\ref{section:conclusion}.

\section{Preliminaries}\label{section:prelim}

\subsection{The Fair Orientation Problem}

An instance of the fair orientation problem is an ordered pair $(G, u)$ where $G = (V(G), E(G))$ is a graph and $u$ is a collection of $|V(G)|$ utility functions $u_i: \mathcal{P}(E(G)) \to \mathbb{R}_{\leq 0}$. We allow for self-loops. Each vertex (resp.\ edge) corresponds to an agent (resp.\ chore). For each agent $i$, the function $u_i$ represents its preferences by assigning a real number to each subset of the edges. In the case of a single edge $e$, we write $u_i(e)$ in lieu of $u_i(\{e\})$ for brevity. We assume that whenever a vertex $i$ is not incident to an edge $e$, the edge $e$ has zero marginal utility to $i$, i.~e.\ $u_i(S \cup \{e\}) = u_i(S)$ for any set $S \subseteq E(G)$. In addition, we make the following two minimal assumptions on each utility function $u_i$:
\begin{itemize}
	\item $u_i$ is {\em monotone}, i.~e., $u_i(S) \leq u_i(T)$ whenever $T \subseteq S$.
	\item If $u_i(e) = 0$ for each $e \in S$, then $u_i(S) = 0$.
\end{itemize}

An {\em allocation} $\pi = (\pi_1, \pi_2, \dots, \pi_{|V(G)|})$ is an ordered collection of $|V(G)|$ pairwise disjoint subsets of $E(G)$ whose union is exactly $E(G)$. For each $i$, the subset $\pi_i$ represents the set of edges assigned to $i$. An allocation is called an {\em orientation} if there is a graph orientation $O$ of $G$ such that $e \in \pi_i$ if and only if $e$ is directed toward $i$ in $O$.

\begin{definition}[EF1]
	An allocation (resp.\ orientation) $\pi$ is {\em EF1} if for each pair of agents $i \neq j$, there exists an edge $e \in \pi_i$ such that $u_i(\pi_i \setminus \{e\}) \geq u_i(\pi_j)$.
\end{definition}

We say an agent $i$ {\em strongly envies} an agent $j$ if there exists an edge $e \in \pi_i$ such that $u_i(\pi_i \setminus \{e\}) < u_i(\pi_j)$.

\begin{definition}[EFX$_0$]
	An allocation (resp.\ orientation) $\pi$ is {\em EFX$_0$} if no agent strongly envies another.
\end{definition}

Given $(G, u)$, \textsc{EFX$_0$-Orientation} asks whether $G$ has an EFX$_0$ orientation. We are particularly interested in instances in which every edge $e = \{i, j\}$ is {\em objective}, i.~e.\ $u_i(e) = 0$ if and only if $u_j(e) = 0$. Equivalently, either $e$ has zero marginal utility for both endpoints or $e$ has negative marginal utility for both. If every edge is objective, we call $(G, u)$ an {\em objective} instance of \textsc{EFX$_0$-Orientation}. We use \textsc{EFX$_0$-Orientation-Objective} to denote the \textsc{EFX$_0$-Orientation} problem restricted to objective instances.

When considering objective instances, we refer to edges with zero (resp.\ negative) marginal utility for both endpoints as {\em dummy edges} (resp.\ {\em negative edges}). Moreover, we define a {\em negative component} of $G$ to be a maximal vertex-induced subgraph $K$ such that for any two vertices $i, j$ of $K$, there is a path in $K$ between $i$ and $j$ that only contains negative edges. The negative components of $G$ form a partition of $V(G)$. See Figure~\ref{fig:negative-component} for an example.

\begin{figure}
	\centering
	\begin{tikzpicture}
		\node[circlevertex] (a1) at (0,0) {};
		\node[circlevertex] (a2) at (1,0) {};
		\node[circlevertex] (a3) at (2,0) {};
		\node[circlevertex] (a4) at (1,1) {};
		\node[emptyvertex] (b1) at (2,1) {};
		\node[emptyvertex] (b2) at (3,1) {};
		\node[emptyvertex] (b3) at (3,0) {};
		\node[emptyvertex] (b4) at (4,0) {};
		\draw (a1) to (a2) (a2) to (a3) (a2) to (a4)
		(b1) to (b2) (b2) to (b3) (b3) to (b1) (b3) to (b4);
		
		\draw[dashed] (a1) to (a4) (a3) to (b2) (b2) to (b4) (a4) to (b1) (b1) to (a3);
	\end{tikzpicture}
	\caption{The graph of an objective instance. Solid (resp.\ dashed) edges denote negative (resp.\ dummy) edges. The set of black (resp.\ white) vertices induces one negative component.}
	\label{fig:negative-component}
\end{figure}

\subsection{Related Decision Problems}

Let $H$ be a graph, $P = \{P_1, P_2, \dots, P_k\}$ be a set of pairwise disjoint subsets of $V(H)$, and $D \subseteq V(H)$ be a subset of vertices. The set $P$ is not necessarily a partition of $V(H)$ because not every vertex needs to be in $P$. We say a subset $C \subseteq V(H)$ is a {\em $(P, D)$-vertex cover} of $H$ if $C$ is a vertex cover of $H$ that contains at most one vertex from each $P_\ell \in P$ and no vertex in $D$, i.~e.\ the following conditions hold:
\begin{itemize}
	\item Each edge of $H$ is incident to at least one vertex in $C$;
	\item $|C \cap P_\ell| \leq 1$ for each $P_\ell \in P$; and
	\item $C \cap D = \emptyset$.
\end{itemize}

\begin{definition}[\textsc{PD-Vertex-Cover}]
	Given a tuple $(H, P, D)$ where $H$ is a graph, $P = \{P_1, P_2, \dots, P_k\}$ is a set of pairwise disjoint subsets of $V(H)$, and $D \subseteq V(H)$, \textsc{PD-Vertex-Cover} asks whether there exists a $(P,D)$-vertex cover $C$ of $H$.
\end{definition}

\begin{definition}[\textsc{2SAT}]
	The \textsc{2SAT} problem is a special case of \textsc{SAT} in which each clause consists of one or two literals.
\end{definition}

\begin{definition}[\textsc{Partition}]
	Given a set $S = \{s_1, s_2, \dots, s_k\}$ of natural numbers, \textsc{Partition} asks whether there exists a partition $S = S_a \cup S_b$ such that $\sum_{x \in S_a} x = \sum_{x \in S_b} x$. (Such a partition is called an {\em equipartition}.)
\end{definition}

\section{EF1 Orientations of Graphs}\label{section:EF1}

We give our $O(|V(G)| + |E(G)|)$-time algorithm for finding an EF1 orientation of a graph $G$ if one exists.

\begin{proposition}\label{prop:EF1-graphs}
	An orientation $\pi$ of a graph $G$ is EF1 if and only if each vertex $i$ receives at most one edge of negative utility to it.
\end{proposition}
\begin{proof}
	Suppose $i$ receives two edges $e = \{i, j\}$ and $e' \neq e$ of negative utility to it. Assume $e'$ is a worst chore in $\pi_i$. Since $e \in \pi_i \setminus \{e'\}$, we have $u_i(\pi_i \setminus \{e'\}) \leq u_i(e) < 0 = u_i(\pi_{j})$, so $i$ envies $j$ even after ignoring $e'$ and $\pi$ is not EF1. Otherwise, $i$ receives at most one edge of negative utility to it. Ignoring such an edge alleviates any envy $i$ has, so $\pi$ is EF1.
\end{proof}

\begin{proposition}\label{prop:EF1-augment}
	If $\pi$ is an EF1 orientation of a graph $G$, then introducing a new edge $e = \{i, j\}$ such that $u_i(e) = 0$ and orienting it toward $i$ results in another EF1 orientation. \qed
\end{proposition}

It is straightforward to decide whether $G$ has an EF1 orientation and compute one if it exists using these two propositions.

\begin{theorem}{theorem}\label{thm:mainContributionThreeB}
	There exists a $O(|V(G)|+|E(G)|)$-time algorithm that decides whether a graph $G$ of chores has an EF1 orientation and outputs one if it exists.
\end{theorem}
\begin{proof}
	First, partition the edge set $E(G)$ of $G$ into $E_0 \cup E_1$, where $E_0$ is the set of edges that have zero utility to at least one endpoint, and $E_1$ is the set of edges that have strictly negative utility to both endpoints. Proposition~\ref{prop:EF1-augment} allows us to ignore the edges in $E_0$ without loss of generality. Hence, every edge has strictly negative utility to both endpoints. In this case,  Proposition~\ref{prop:EF1-graphs} implies that an orientation $\pi$ of $G$ is EF1 if and only if each vertex $i$ receives at most one edge. Clearly, this is the case if and only if each component $K$ of $G$ contains at most $|V(K)|$ edges.
	
	To construct an EF1 orientation in the case that it exists, orient each edge $e \in E_0$ toward an endpoint $i$ such that $u_i(e) = 0$. Then, using BFS, find the components of the subgraph of $G$ induced by the edges of $E_1$, and orient the edges of $E_1$ in a way such that each vertex of $G$ receives at most one.
	
	Ignoring the set $E_0$ of edges that have zero utility to at least one endpoint takes $O(|E(G)|)$ time by checking every edge. Verifying the number of edges in each component of $G$ takes $O(|V(G)| + |E(G)|)$ time using BFS. Constructing the EF1 orientation also takes $O(|V(G)| + |E(G)|)$ time because it can be done using BFS.
\end{proof}

\section{EFX$_0$ Orientations of Graphs}\label{section:graphs}

In this section, we give our $O((|V(G)|+|E(G)|)^2)$-time algorithm \textsc{FindEFXOrientation} (Algorithm~\ref{alg:FindEFXOrientation}) for \textsc{EFX$_0$-Orientation}. Our algorithm is essentially a reduction from \textsc{EFX$_0$-Orientation} to \textsc{2SAT}. It comprises three nested routines, each functioning as a reduction from one decision problem to another. See Figure~\ref{fig:schematic} for a visual overview.

We present the proofs of correctness of the algorithms in their nested order, starting with innermost one, \textsc{FindPDVertexCover} (Algorithm~\ref{alg:FindPDVertexCover}).

\begin{figure}
	\centering
	\begin{tikzpicture}
		\node[sqvertex] (P1) at (0,3*1.3) {\textsc{EFX$_0$-Orientation}};
		\node[sqvertex] (P2) at (0,2*1.3) {\textsc{EFX$_0$-Orientation-Objective}};
		\node[sqvertex] (P3) at (0,1*1.3) {\textsc{PD-Vertex-Cover}};
		\node[sqvertex] (P4) at (0,0) {\textsc{2SAT}};
		
		\draw[arc] (P1) to node[right] {\textsc{FindEFXOrientation} (Algorithm~\ref{alg:FindEFXOrientation})} (P2);
		
		\draw[arc] (P2) to node[right] {\textsc{FindEFXOrientationObj} (Algorithm~\ref{alg:FindEFXOrientObj})} (P3);
		
		\draw[arc] (P3) to node[right] {\textsc{FindPDVertexCover} (Algorithm~\ref{alg:FindPDVertexCover})} (P4);
	\end{tikzpicture}
	\caption{The reductions our algorithms represent. Each node on the left is a decision problem. A directed edge between nodes is a reduction and is labelled with the relevant algorithm. \vspace{\baselineskip}}
	\label{fig:schematic}
\end{figure}

\subsection{\textsc{FindPDVertexCover} (Algorithm~\ref{alg:FindPDVertexCover})}

\textsc{FindPDVertexCover} (Algorithm~\ref{alg:FindPDVertexCover}) accepts an instance $(H, P, D)$ of \textsc{PD-Vertex-Cover} as input, and outputs a $(P, D)$-vertex cover of $H$ if one exists and \texttt{false} otherwise.

Recall that in an instance $(H, P, D)$ of \textsc{PD-Vertex-Cover}, $H$ is a graph, $P = \{P_1, P_2, \dots, P_k\}$ is a set of pairwise disjoint subsets of $V(H)$, and $D \subseteq V(H)$. The algorithm starts by constructing an instance $\phi$ of \textsc{2SAT} on Line~\ref{line:1.create-sat}, which contains a Boolean variable $x_i$ for each vertex $i$ of $H$ and three types of clauses:
\begin{itemize}
	\item (type 1) the clause $\{x_i, x_j\}$ for each edge $e = \{i, j\}$ of $H$ (possibly $i=j$),
	\item (type 2) the clause $\{\overline{x_i}, \overline{x_j}\}$ for each set $P_\ell \in P$ and each pair of distinct vertices $i, j \in P_\ell$, and
	\item (type 3) the clause $\{\overline{x_i}\}$ for each vertex $i \in D$.
\end{itemize}
Then, the algorithm determines whether $\phi$ has a satisfying assignment $f$ on Line~\ref{line:1.call-sat}. If so, the algorithm outputs the $(P,D)$-vertex cover $C \coloneqq \{i \in V(H) \mid f(x_i) = \texttt{true}\}$. Otherwise, the algorithm outputs \texttt{false}.

\begin{algorithm}
	\caption{$\textsc{FindPDVertexCover}(H,P,D)$}
	\label{alg:FindPDVertexCover}
	
	\hspace*{\algorithmicindent} \textbf{Input:} An instance $(H, P, D)$ of \textsc{PD-Vertex-Cover}.
	
	\hspace*{\algorithmicindent} \textbf{Output:} A $(P,D)$-vertex cover $C$ of $H$ if one exists and \texttt{false} otherwise.
	\begin{algorithmic}[1] 
		
		
		
		
		
		
		
		\State $\phi \gets$ the 2SAT instance defined above \label{line:1.create-sat}
		\State $f \gets \textsc{2SAT}(\phi)$\label{line:1.call-sat}
		
		\If{$f = \texttt{false}$} \Comment{$\phi$ has no satisfying assignment}
		\State \Return \texttt{false}\label{line:1.return-false}
		\EndIf

		\State \Return $C \coloneqq \{i \in V(H) \mid f(x_i) = \texttt{true}\}$\label{line:1.return-C}
	\end{algorithmic}
\end{algorithm}

\begin{lemma}\label{lemma:reduc:FindPDVertexCover}
	The following statements hold for \textsc{FindPDVertexCover} (Algorithm~\ref{alg:FindPDVertexCover}):
	\begin{enumerate}
		\item If $H$ has a $(P,D)$-vertex cover, then $\phi$ has a satisfying assignment.
		
		\item If $\phi$ has a satisfying assignment $f$, then $C \coloneqq \{i \in H \mid f(x_i) = \normalfont{\texttt{true}}\}$ is a $(P,D)$-vertex cover of $H$. Otherwise, \textsc{FindPDVertexCover} outputs \texttt{false}.
		
		\item \textsc{FindPDVertexCover} runs in $O(|V(H)|^2)$ time.
	\end{enumerate}
\end{lemma}
\begin{proof}
	(1): Suppose $H$ has a $(P, D)$-vertex cover $C$. We claim the truth assignment $f$ in which $f(x_i) = \texttt{true}$ if and only if $i \in C$ to be a satisfying assignment for $\phi$. Since $C$ is a $(P, D)$-vertex cover of $H$, each edge $e$ of $H$ is incident to at least one vertex in $C$. In other words, for each edge $e = \{i, j\}$ of $H$, we have $i \in C$ or $j \in C$. Hence, at least one of $x_i, x_j$ is true under $f$, so the type 1 clause $\{x_i, x_j\}$ is satisfied.
	
	On the other hand, the definition of $(P, D)$-vertex covers ensures $|C \cap P_\ell| \leq 1$ for each $P_\ell \in P$. Hence, each $P_\ell$ contains at most one vertex $i$ such that $f(x_i) = \texttt{true}$, so all type 2 clauses are satisfied.
	
	Finally, because $C$ is a $(P, D)$-vertex cover, we have $C \cap D = \emptyset$. Hence, for each $i \in D$, we have $i \notin C$, so $f(x_i) = \texttt{false}$, ensuring all type 3 clauses are satisfied.
	
	(2): Suppose $\phi$ has a satisfying assignment $f$. We claim that $C \coloneqq \{i \in V(H) \mid f(x_i) = \texttt{true}\}$ is a $(P, D)$-vertex cover of $H$. Since the type 1 clause $\{x_i, x_j\}$ is satisfied under $f$, the edge $e = \{i, j\}$ of $H$ is incident to at least one vertex of $C$, namely, $i$ or $j$. Since the type 2 clauses are satisfied, for each fixed $P_\ell \in P$, there is at most one vertex $i \in P_\ell$ such that $x_i$ is true under $f$. Hence, $|C \cap P_\ell| \leq 1$. Finally, the type 3 clauses ensure that $f(x_i) = \texttt{false}$ for each vertex $i \in D$, so $i \notin C$. Thus, $C \cap D = \emptyset$.
	
	Otherwise, $\phi$ has no satisfying assignment and \textsc{FindPDVertexCover} outputs \texttt{false} on Line~\ref{line:1.return-false}.
	
	(3): Constructing $\phi$ on Line~\ref{line:1.create-sat} takes $O(|V(H)|^2)$ time because $\phi$ contains $|V(H)|$ variables and $O(|V(H)|^2)$ clauses. Verifying whether $\phi$ has a satisfying assignment on Line~\ref{line:1.call-sat} takes time linear in the number of variables and clauses using the algorithm due to \citet{aspvall1979linear}, which is $O(|V(H)|^2)$. Finally, constructing $C$ takes $O(|V(H)|)$ time by checking each vertex of $H$. Thus, \textsc{FindPDVertexCover} runs in $O(|V(H)|^2)$ time.
\end{proof}

\subsection{\textsc{FindEFXOrientObj} (Algorithm~\ref{alg:FindEFXOrientObj})}

\textsc{FindEFXOrientObj} (Algorithm~\ref{alg:FindEFXOrientObj}) takes as input an instance $(G^o, u^o)$ of \textsc{EFX$_0$-Orientation-Objective} and outputs an EFX$_0$ orientation of $G^o$ if one exists and \texttt{false} otherwise. Recall that in an objective instance, all edges are objective, and we distinguish dummy edges that provide zero marginal utility to both endpoints and negative edges that provide negative marginal utility to both endpoints. Recall also that a negative component of $G^o$ is a maximal vertex-induced subgraph $K$ such that for any two vertices $i, j$ of $K$, there exists a path in $K$ between $i$ and $j$ that only contains negative edges.

We make the following crucial observation.

\begin{proposition}\label{prop:one-neg-or-all-dummies}
	Let $(G^o, u^o)$ be an instance of \textsc{EFX$_0$-Orientation-Objective}. An orientation $\pi^o$ of $G^o$ is EFX$_0$ if and only for each vertex $i$, the orientation $\pi^o$ contains a unique edge directed toward $i$ or every edge directed toward $i$ is a dummy edge.
\end{proposition}
\begin{proof}
	Let $\pi^o = (\pi_1^o, \pi_2^o, \dots, \pi_{|V(G^o)|}^o)$ be an orientation of $G^o$ and suppose that for each vertex $i$, the orientation $\pi^o$ contains a unique edge directed toward $i$ or every edge directed toward $i$ is a dummy edge. If $\pi^o$ contains a unique edge $e$ directed toward $i$, then we have $u_i^o(\pi^o_i \setminus \{e\}) = u_i^o(\emptyset) = 0 \geq u_i^o(\pi_j^o)$ for any agent $j$. Otherwise, every edge directed toward $i$ is a dummy edge, so $u_i^o(\pi_i^o) = 0 \geq u_i^o(\pi_j^o)$ for any $j$. Hence, no agent strongly envies another and $\pi^o$ is EFX$_0$.
	
	Conversely, suppose the orientation $\pi^o$ contains two edges $e, e'$ directed toward some $i$ and not every edge directed toward $i$ is a dummy edge. Without loss of generality, assume $u_i^o(e') < 0$. Let $j, j'$ denote the other endpoints of the edges $e, e'$, respectively. Note that $j \neq j'$ because $G$ does not contain parallel edges. Also, $u_i^o(\pi_j^o) = 0$ because $i$ receives the edge $e$ between $i$ and $j$. Since $u_i^o(e') < 0$ and $\{e'\} \subseteq \pi_i^o \setminus \{e\}$, the monotonicity of $u_i^o$ implies $u_i^o(\pi_i^o \setminus \{e\}) \leq u_i^o(e') < 0 = u_i^o(\pi_j^o)$. So, $\pi^o$ is not EFX$_0$.
\end{proof}

Proposition~\ref{prop:one-neg-or-all-dummies} gives a graph theoretical condition that is equivalent to the EFX$_0$ condition, which \textsc{FindEFXOrientObj} exploits. The algorithm first finds the set $\mathcal{K}$ of negative components of $G^o$ on Line~\ref{line:2.find-neg-comps}. If some negative component $K$ contains more than $|V(K)|$ negative edges, the algorithm concludes that $G^o$ has no EFX$_0$ orientation and outputs \texttt{false}. Otherwise, every negative component $K$ has at most $|V(K)|$ negative edges, and the algorithm constructs the instance $(H, P, D)$ of \textsc{PD-Vertex-Cover} on Line~\ref{line:2.construct-HPD} defined as:
\begin{itemize}
	\item $H$ is the graph with the same vertex set as $G^o$ and the set of dummy edges of $G^o$;
	\item $P = \{K \in \mathcal{K} \mid K \text{ has exactly } |V(K)|-1 \text{ negative edges}\}$;
	\item $D$ is the set of vertices that belong to a negative component $K \in \mathcal{K}$ that has exactly $|V(K)|$ negative edges.
\end{itemize}

Using \textsc{FindPDVertexCover} (Algorithm~\ref{alg:FindPDVertexCover}) on Line~\ref{line:2.subroutine}, the algorithm tries to find a $(P,D)$-vertex cover $C$ of $H$. If successful, it uses $C$ to construct the orientation $\pi^o$ of $G^o$  defined below and outputs it. Otherwise, it outputs \texttt{false}.

We give the construction of the orientation $\pi^o$ of $G^o$. Because $C$ is a $(P, D)$-vertex cover of $H$, each edge $e$ of $H$ (equivalently, each dummy edge of $G^o$) is incident to a vertex $i \in C$. Orient $e$ toward $i$, with arbitrary tie-breaking when two such $i$ exist. We then orient the negative edges of $G^o$ using Observation~\ref{obs:graph}, which implies $G^o$ has an orientation such that every vertex has in-degree at most one (i.~e.\ each vertex receives at most one negative edge) because every negative component $K$ of $G^o$ has at most $|V(K)|$ negative edges.

\begin{algorithm}
	\caption{$\textsc{FindEFXOrientObj}(G^o,u^o)$}
	\label{alg:FindEFXOrientObj}
	
	\hspace*{\algorithmicindent} \textbf{Input:} An instance $(G^o, u^o)$ of \textsc{EFX$_0$-Orientation-Objective}.
	
	\hspace*{\algorithmicindent} \textbf{Output:}
	An EFX$_0$ orientation $\pi^o$ of $G^o$ if one exists and \texttt{false} otherwise.
	\begin{algorithmic}[1] 
		\State $\mathcal{K} \gets \{K \mid K \text{ is a negative component of } G^o \}$  \label{line:2.find-neg-comps}
		
		
		
		
		\If{$\exists K \in \mathcal{K} \mid K$ contains $>|V(K)|$ negative edges}
		\State \Return \texttt{false} \label{line:2.return-false-1}
		\EndIf
		
		\State $(H, P, D) \gets $ the instance of \textsc{PD-Vertex-Cover} defined above \label{line:2.construct-HPD}
		
		
		
		\State $C \gets \textsc{FindPDVertexCover}(H,P,D)$ \label{line:2.subroutine} \Comment{Algorithm~\ref{alg:FindPDVertexCover}}
		
		\If{$C = \texttt{false}$}
		\State \Return \texttt{false} \label{line:2.return-false-2}
		\EndIf
		
		\State \Return the orientation $\pi^o$ of $G^o$ defined above \label{line:2.construct-pio} \label{line:2.return}
	\end{algorithmic}
\end{algorithm}

\begin{lemma}\label{lemma:reduc:FindEFXOrientObj}
	The following statements hold for \textsc{FindEFXOrientObj} (Algorithm~\ref{alg:FindEFXOrientObj}):
	\begin{enumerate}
		\item If $G^o$ has an EFX$_0$ orientation, then no negative component $K$ of $G^o$ contains $>|V(K)|$ negative edges and $H$ has a $(P,D)$-vertex cover.
		\item If $C$ is a $(P,D)$-vertex cover of $H$, then $\pi^o$ is an EFX$_0$ orientation of $G^o$.
		\item \textsc{FindEFXOrientObj} runs in $O(|V(G^o)|^2)$ time.
	\end{enumerate}
\end{lemma}
\begin{proof}
	(1:) Suppose $G^o$ has an EFX$_0$ orientation $\pi^o$. First, we show that no negative component $K$ of $G^o$ contains $>|V(K)|$ negative edges. Fix a negative component $K$. By Proposition~\ref{prop:one-neg-or-all-dummies}, for each vertex $i$, the orientation $\pi^o$ contains a unique edge directed toward $i$ or every edge directed toward $i$ is a dummy edge. In particular, this is true for each vertex in $K$, so $\pi^o$ contains at most $|V(K)|$ negative edges directed toward a vertex of $K$. On the other hand, every negative edge of $K$ is directed toward a vertex of $K$ because $\pi^o$ is an orientation. It follows that $K$ contains at most $|V(K)|$ negative edges.
	
	We claim that the set $C$ of vertices $i$ of $G^o$ such that $\pi^o$ contains a dummy edge directed toward $i$ is a $(P, D)$-vertex cover of $H$. First we show that each edge of $H$ is incident to at least one vertex in $C$. Let $e$ be an edge of $H$. By the definition of $H$, the edge $e$ is a dummy edge of $G^o$. Let $i$ be the endpoint toward which $e$ is directed in the orientation $\pi^o$ of $G^o$. The definition of $C$ implies $i \in C$, so $e$ is incident to at least one vertex in $C$.
	
	Next, we show that $|C \cap K_\ell| \leq 1$ for each $K_\ell \in P$. Fix any negative component $K_\ell$ of $G^o$. Since $K_\ell$ is a negative component, it contains at least $|K_\ell|-1$ negative edges. Since $\pi^o$ is EFX$_0$, Proposition~\ref{prop:one-neg-or-all-dummies} implies that for each vertex $i$, the orientation $\pi^o$ contains a unique edge directed toward $i$ or every edge directed toward $i$ is a dummy edge. Since $K_\ell$ contains at least $|K_\ell|-1$ negative edges and each vertex of $K$ can have at most one negative edge directed towards it, the pigeonhole principle implies that $K_\ell$ contains at most one vertex to which dummy edges are directed in $\pi^o$. In other words, at most one vertex of $K_\ell$ is in $C$, i.~e., $|C \cap K_\ell| \leq 1$.
	
	Finally, we show that $C \cap D = \emptyset$. Suppose $i \in D$. Because $i \in D$, the negative component $K$ to which $i$ belongs contains $|V(K)|$ negative edges. Proposition~\ref{prop:one-neg-or-all-dummies} implies that for each vertex $j$, the orientation $\pi^o$ contains a unique edge directed toward $j$ or every edge directed toward $j$ is a dummy edge. In particular, each vertex of $K$ has at most one negative edge directed toward it. By the pigeonhole principle, each vertex of $K$ has exactly one negative edge of $K$ directed toward it. In particular, this is true for $i$, so Proposition~\ref{prop:one-neg-or-all-dummies} implies no dummy edge is directed toward $i$. Thus, $i \notin C$. It follows that $C \cap D = \emptyset$ and $H$ indeed has a $(P,D)$-vertex cover.
	
	(2:) Suppose the algorithm finds a $(P, D)$-vertex cover $C$ of $H$ on Line~\ref{line:2.subroutine}. First, we show that $\pi^o$ indeed orients every edge of $G^o$. Recall that $H$ is the graph with the same vertex set as $G^o$ and the set of dummy edges of $G^o$. Since $C$ is a $(P, D)$-vertex cover of $H$, each edge of $H$ (thereby each dummy edge of $G^o$) is incident to at least one vertex in $C$. It follows that every dummy edge of $G^o$ is oriented by $\pi^o$. On the other hand, let $e$ be a negative edge of $G^o$ and $K$ be the negative component containing $e$. Since the algorithm successfully finds a $(P, D)$-vertex cover $C$ of $H$ on Line~\ref{line:2.subroutine}, it does not terminate on Line~\ref{line:2.return-false-1}. So, $K$ contains exactly $|V(K)|$ or $|V(K)|-1$ negative edges, and $\pi^o$ orients all of them. It follows that $\pi^o$ indeed  orients every edge of $G^o$.
	
	It remains to show $\pi^o$ to be EFX$_0$. We verify the condition in Proposition~\ref{prop:one-neg-or-all-dummies} for every vertex $i$. Let $C' \subseteq C$ be the set of vertices with a dummy edge directed toward it in $\pi^o$ and consider any $i \in C'$. Since $C$ is a $(P, D)$-vertex cover of $H$, we have $C \cap D = \emptyset$, so $i \notin D$. Because $D$ is the set of vertices that belong to a negative component $K$ of $G^o$ that has exactly $|V(K)|$ negative edges, $i$ belongs to a negative component $K$ of $G^o$ containing exactly $|V(K)|-1$ negative edges. As discussed while defining $\pi^o$, at most one vertex of such a negative component $K$ could have received a dummy edge, which clearly is the vertex $i$. So, $\pi^o$ orients the negative edges of $K$ in the direction away from $i$, so that they form a spanning tree of $K$ rooted at $i$. Hence, every edge directed toward $i$ is a dummy edge. it follows that the condition in Proposition~\ref{prop:one-neg-or-all-dummies} holds for every $i \in C'$.
	
	On the other hand, the definition of $\pi^o$ ensures each vertex receives at most one negative edge. In particular, each vertex $i \notin C'$ receives no dummy edges and at most one negative edge, so the condition in Proposition~\ref{prop:one-neg-or-all-dummies} holds for every $i \notin C'$ as well. By Proposition~\ref{prop:one-neg-or-all-dummies}, $\pi^o$ is EFX$_0$.
	
	(3:) Finding the set $\mathcal{K}$ of negative components of $G^o$ on Line~\ref{line:2.find-neg-comps} can be done using a BFS using $O(|V(G^o)| + |E(G^o)|)$ time. Constructing $(H,P,D)$ on Lines~\ref{line:2.construct-HPD} requires counting the number of negative edges in each negative component $K \in \mathcal{K}$, which can be done in $O(|E(G^o)|)$ time. By Lemma~\ref{lemma:reduc:FindPDVertexCover}(3), \textsc{FindPDVertexCover} on Line~\ref{line:2.subroutine} runs in $O(|V(H)|^2)$ time.
	
	The construction of $\pi^o$ on Line~\ref{line:2.construct-pio} is done in two steps. First, orienting the dummy edges can be done in $O(|V(G^o)| + |E(G^o)|)$ time because there it requires checking which endpoint of each dummy edge is in the $(P,D)$-vertex cover $C$. Then, orienting the negative edges can be done in $O(|V(G^o)|^2)$ time according to Observation~\ref{obs:graph}.
	
	By summing the above and observing that $|V(H)| = |V(G^o)|$, we conclude that \textsc{FindEFXOrientObj} runs in $O(|V(G^o)|^2)$ time.
\end{proof}

\begin{lemma}\label{lemma:FindEFXOrientObj}
	Given an instance $(G^o, u^o)$ of \textsc{EFX$_0$-Orientation-Objective}, \textsc{FindEFXOrientObj} outputs an EFX$_0$ orientation $\pi^o$ of $G^o$ if one exists and \normalfont{\texttt{false}} otherwise in $O(|V(G^o)|^2)$ time.
\end{lemma}
\begin{proof}
	If $G^o$ has an EFX$_0$ orientation, then no negative component $K$ of $G^o$ contains $>|V(K)|$ negative edges and $H$ has a $(P,D)$-vertex cover by Lemma~\ref{lemma:reduc:FindEFXOrientObj}(1). So, Line~\ref{line:2.subroutine} finds a $(P, D)$-vertex cover $C$ of $H$ by Lemma~\ref{lemma:reduc:FindPDVertexCover} and \textsc{FindEFXOrientObj} outputs the EFX$_0$ orientation $\pi^o$ by Lemma~\ref{lemma:reduc:FindEFXOrientObj}(2).
	
	Otherwise, $G^o$ does not have an EFX$_0$ orientation. If some negative component $K$ of $G^o$ contains $>|V(K)|$ negative edges, then \textsc{FindEFXOrientObj} outputs \texttt{false} on Line~\ref{line:2.return-false-1}, Assume every negative component $K$ of $G^o$ contains at most $|V(K)|$ negative edges, so \textsc{FindEFXOrientObj} constructs the instance $(H, P, D)$ of \textsc{PD-Vertex-Cover} on Line~\ref{line:2.construct-HPD} and attempts to find a $(P, D)$-vertex cover $C$ of $H$ on Line~\ref{line:2.subroutine}. By the contrapositive of Lemma~\ref{lemma:reduc:FindEFXOrientObj}(2), such a $C$ does not exist, so \textsc{FindEFXOrientObj} outputs \texttt{false} on Line~\ref{line:2.return-false-2}.
\end{proof}

\subsection{\textsc{FindEFXOrientation} (Algorithm~\ref{alg:FindEFXOrientation})}

\textsc{FindEFXOrientation} (Algorithm~\ref{alg:FindEFXOrientation}) accepts an instance $(G,u)$ of \textsc{EFX$_0$-Orientation} as input and outputs an EFX$_0$ orientation of $G$ if one exists and \texttt{false} otherwise.

It first constructs the instance $(G^o, u^o)$ of \textsc{EFX$_0$-Orientation-Objective} by subdividing each non-objective edge of $G$. Specifically, for each non-objective edge $e_{ij} = \{i, j\}$, first assume $u_i(e_{ij}) = 0$ without loss of generality by possibly exchanging the labels of $i$ and $j$, then replace $e_{ij}$ by a new vertex $k$ and two new edges $e_{ik} = \{i, k\}, e_{jk} = \{j, k\}$. The utilities of the new edges are defined as shown in Figure~\ref{fig:subdivide}, in which $\beta \coloneqq u_j(e_{ij})$. The objective edges and their utilities remain unchanged.

It then uses \textsc{FindEFXOrientObj} (Algorithm~\ref{alg:FindEFXOrientObj}) on Line~\ref{line:3.subroutine} to try to find an EFX$_0$ orientation of $G^o$. If none exists, it outputs \texttt{false} on Line~\ref{line:3.return-false}. Otherwise, an EFX$_0$ orientation $\pi^o$ of $G^o$ is successfully found. Using $\pi^o$, the algorithm constructs and outputs the orientation $\pi$ of $G$ defined as follows.

Let $e_{ij}$ be an edge of $G$. If $e_{ij}$ was not subdivided in the construction of $G^o$, then it is also an edge in $G^o$. In this case, $\pi$ orients $e_{ij}$ in $G$ in the same way as $\pi^o$ orients $e_{ij}$ in $G^o$.

Otherwise, $e_{ij}$ was subdivided in the construction of $G^o$, and corresponds to two edges $e_{ik}, e_{jk}$ of $G^o$, where $u_i^o(e_{ik}) = u_i(e_{ij}) = 0$. In this case, $\pi$ orients $e_{ij}$ in $G$ in the same direction as $\pi^o$ orients $e_{ik}$ in $G^o$ (with respect to Figure~\ref{fig:subdivide}). That is, if $\pi^o$ orients $e_{ik}$ toward $i$ in $G^o$, then $\pi$ orients $e_{ij}$ toward $i$ in $G$. Otherwise, $\pi^o$ orients $e_{ik}$ toward $k$ in $G^o$, and $\pi$ orients $e_{ij}$ toward $j$ in $G$.

\begin{figure}
	\centering
	\subfigure[the edge $e_{ij}$ in $G$]{
		\centering
		\begin{tikzpicture}
			\node[blkvertex] (i) at (0, 0) {};
			\node[blkvertex] (j) at (1.5, 0) {};
			
			\node (li) at (0, -0.25) {$i$};
			\node (lj) at (1.5, -0.25) {$j$};
			
			\node (l1) at (0.2, 0.2) {$0$};
			\node (l2) at (1.3, 0.2) {$\beta$};
			
			\draw (i) to node [below] {$e_{ij}$} (j);
	\end{tikzpicture}}
	\hspace{1cm}
	\subfigure[after subdivision in $G^o$]{
		\centering
		\begin{tikzpicture}
			\node[blkvertex] (i) at (0, 0) {};
			\node[blkvertex] (j) at (3, 0) {};
			\node[blkvertex] (k) at (1.5, 0) {};
			
			\node (li) at (0, -0.25) {$i$};
			\node (lk) at (1.5, -0.25) {$k$};
			\node (lj) at (3, -0.25) {$j$};
			
			\node (l1) at (0.2, 0.2) {$0$};
			\node (l2) at (1.3, 0.2) {$0$};
			\node (l3) at (1.7, 0.2) {$\beta$};
			\node (l4) at (2.8, 0.2) {$\beta$};
			
			\draw (i) to node [below] {$e_{ik}$} (k);
			\draw (k) to node [below] {$e_{jk}$} (j);
	\end{tikzpicture}}
	
	\caption{Subdivision of $e_{ij}$ during the construction of $G^o$ by \textsc{FindEFXOrientation}. The labels above an edge indicate the utility the edge has to its two endpoints. We write $\beta \coloneqq u_j(e_{ij})$ for clarity.}
	\label{fig:subdivide}
\end{figure}

\begin{algorithm}
	\caption{$\textsc{FindEFXOrientation}(G,u)$}
	\label{alg:FindEFXOrientation}
	
	\hspace*{\algorithmicindent} \textbf{Input:} An instance $(G, u)$ of \textsc{EFX$_0$-Orientation}.
	
	\hspace*{\algorithmicindent} \textbf{Output:} An EFX$_0$ orientation $\pi$ of $G$ if one exists and \texttt{false} otherwise.
	
	\begin{algorithmic}[1] 
		\State $(G^o, u^o) \gets$ the instance of \textsc{EFX$_0$-Orientation-Objective} defined above

		\State $\pi^o \gets \textsc{FindEFXOrientObj}(G^o, u^o)$\label{line:3.subroutine} \Comment{Algorithm~\ref{alg:FindEFXOrientObj}}
		
		\If{$\pi^o = \texttt{false}$}
		\State \Return \texttt{false}\label{line:3.return-false}
		\EndIf
		
		\State \Return $\pi \coloneqq$ the orientation of $G$ defined above \label{line:3.return-pi}
		
	\end{algorithmic}
\end{algorithm}

\begin{lemma}\label{lemma:reduc:FindEFXOrientation}
	The following statements hold for \textsc{FindEFXOrientation} (Algorithm~\ref{alg:FindEFXOrientation}):
	\begin{enumerate}
		\item If $G$ has an EFX$_0$ orientation, then so does $G^o$.
		\item If $\pi^o$ is an EFX$_0$ orientation of $G^o$, then $\pi$ is an EFX$_0$ orientation of $G$.
	\end{enumerate}
\end{lemma}
\begin{proof}
	(1:) Suppose $\pi$ is an EFX$_0$ orientation of $G$. We construct the orientation $\pi^o$ of $G^o$ as follows. For each edge $e$ of $G$ that was not subdivided during the construction of $G^o$, we let $\pi^o$ orient it in $G^o$ in the same direction that $\pi$ orients it in $G$. On the other hand, suppose $e_{ik}$ and $e_{jk}$ are a pair of edges of $G^o$ that resulted from the subdivision of the edge $e_{ij}$ of $G$. If $\pi$ orients $e_{ij}$ leftward in terms of Figure~\ref{fig:subdivide} (resp.\ rightward), then $\pi^o$ orients both $e_{ik}$ and $e_{jk}$ leftward (resp.\ rightward).
	
	We claim $\pi^o$ is EFX$_0$. We refer to a vertex that is shared between both $G$ and $G^o$ as a {\em real} vertex, and a vertex that exists only in $G^o$ as a {\em fake} vertex. Fake vertices result from the subdivisions during the construction of $G^o$. Clearly, from the local perspective of each real vertex, the orientations $\pi$ and $\pi^o$ appear the same. Since no real vertex strongly envies a neighbour in $\pi$, no real vertex strongly envies a neighbour in $\pi^o$. On the other hand, the construction of $\pi^o$ ensures each fake vertex receives exactly one incident edge, so no fake vertex strongly envies a neighbour in $\pi^o$. Thus, $\pi^o$ is EFX$_0$.
	
	(2:) Suppose $\pi^o$ is an EFX$_0$ orientation of $G^o$. Proposition~\ref{prop:one-neg-or-all-dummies} implies for each vertex $i$, the orientation $\pi^o$ contains a unique edge directed toward $i$ or every edge directed toward $i$ is a dummy edge. In particular, for any fake vertex $k$ in $G^o$ and the two edges $e_{ik}, e_{jk}$ incident to it, $e_{ik}$ is a dummy edge and $e_{jk}$ is a negative edge, so at most one of them is directed toward $k$ in $\pi^o$. Hence, there are only three possibilities for how pairs of fake edges are oriented in $\pi^o$ (see Figure~\ref{fig:construct-pi}).
	
	Consider the first two cases depicted in Figure~\ref{fig:construct-pi}. In Figure~\ref{fig:construct-pi}(a), both $e_{ik}$ and $e_{jk}$ are directed rightward in $\pi^o$, and $\pi$ orients the corresponding edge $e_{ij}$ rightward. In Figure~\ref{fig:construct-pi}(b), both $e_{ik}$ and $e_{jk}$ are directed leftward in $\pi^o$, and $\pi$ orients $\pi_{ij}$ leftward. In either case, the real vertices $i$ and $j$ see no difference between $\pi$ and $\pi^o$ locally in terms of the utilities they derive from the edges $e_{ij}, e_{ik}, e_{jk}$ between $i$ and $j$ that they receive, because $u_i(e_{ij}) = u_i(e_{ik})$ and $u_j(e_{ij}) = u_j(e_{jk})$.
	
	In the case depicted by Figure~\ref{fig:construct-pi}(c), the orientation $\pi$ directs the edge $e_{ij}$ toward $i$. So, $i$ sees no difference between $\pi$ and $\pi^o$ locally between $i$ and $j$, and $j$ thinks that it has given an edge of utility $u_j(e_{jk})$ to $i$ by going from $\pi^o$ to $\pi$.
	
	Hence, for each pair of adjacent vertices $i, j$ of $G$, the vertex $i$, having possibly given some of its edges in $\pi^o$ to $j$ going from $\pi^o$ to $\pi$, does not strongly envy $j$ in $\pi$ because it does not strongly envy $j$ in $\pi^o$. Thus, $\pi$ is EFX$_0$.
\end{proof}

\begin{figure}
	\centering
	\subfigure[]{
		\centering
		\begin{tikzpicture}
			\node[blkvertex] (i) at (0, 0) {};
			\node[blkvertex] (j) at (3, 0) {};
			\node[blkvertex] (k) at (1.5, 0) {};
			
			\node (li) at (0, -0.25) {$i$};
			\node (lk) at (1.5, -0.25) {$k$};
			\node (lj) at (3, -0.25) {$j$};
			
			\node (l1) at (0.2, 0.2) {$0$};
			\node (l2) at (1.3, 0.2) {$0$};
			\node (l3) at (1.7, 0.2) {$\beta$};
			\node (l4) at (2.8, 0.2) {$\beta$};
			
			\draw[arc] (i) to node [below] {$e_{ik}$} (k);
			\draw[arc] (k) to node [below] {$e_{jk}$} (j);
		\end{tikzpicture}}
	\subfigure[]{
		\centering
		\begin{tikzpicture}
			\node[blkvertex] (i) at (0, 0) {};
			\node[blkvertex] (j) at (3, 0) {};
			\node[blkvertex] (k) at (1.5, 0) {};
			
			\node (li) at (0, -0.25) {$i$};
			\node (lk) at (1.5, -0.25) {$k$};
			\node (lj) at (3, -0.25) {$j$};
			
			\node (l1) at (0.2, 0.2) {$0$};
			\node (l2) at (1.3, 0.2) {$0$};
			\node (l3) at (1.7, 0.2) {$\beta$};
			\node (l4) at (2.8, 0.2) {$\beta$};
			
			\draw[arc] (k) to node [below] {$e_{ik}$} (i);
			\draw[arc] (j) to node [below] {$e_{jk}$} (k);
		\end{tikzpicture}}
	\subfigure[]{
		\centering
		\begin{tikzpicture}
			\node[blkvertex] (i) at (0, 0) {};
			\node[blkvertex] (j) at (3, 0) {};
			\node[blkvertex] (k) at (1.5, 0) {};
			
			\node (li) at (0, -0.25) {$i$};
			\node (lk) at (1.5, -0.25) {$k$};
			\node (lj) at (3, -0.25) {$j$};
			
			\node (l1) at (0.2, 0.2) {$0$};
			\node (l2) at (1.3, 0.2) {$0$};
			\node (l3) at (1.7, 0.2) {$\beta$};
			\node (l4) at (2.8, 0.2) {$\beta$};
			
			\draw[arc] (k) to node [below] {$e_{ik}$} (i);
			\draw[arc] (k) to node [below] {$e_{jk}$} (j);
		\end{tikzpicture}}
	
	\caption{The three possibilities involving the fake edges $e_{ik}, e_{jk}$ in the EFX$_0$ orientation $\pi^o$ of $G^o$.}
	\label{fig:construct-pi}
\end{figure}

\begin{theorem}\label{thm:main}
	There exists a $O((|V(G)|+|E(G)|)^2)$-time algorithm that decides whether a graph $G$ of chores has an EFX$_0$ orientation and outputs one if it exists.
\end{theorem}
\begin{proof}
	We show \textsc{FindEFXOrientation} to be such an algorithm. If $G$ has an EFX$_0$ orientation, then $G^o$ does as well by Lemma~\ref{lemma:reduc:FindEFXOrientation}(1), so the subroutine \textsc{FindEFXOrientObj} produces such an orientation $\pi^o$ of $G^o$ by Lemma~\ref{lemma:FindEFXOrientObj}. In this case, \textsc{FindEFXOrientation} outputs $\pi$, which is an EFX$_0$ orientation of $G$ by Lemma~\ref{lemma:reduc:FindEFXOrientation}(2). Otherwise, $G$ does not have an EFX$_0$ orientation, so so $G^o$ has no EFX$_0$ orientation by Lemma~\ref{lemma:reduc:FindEFXOrientation}(2), causing the subroutine \textsc{FindEFXOrientObj} on Line~\ref{line:3.subroutine} to output \texttt{false} and \textsc{FindEFXOrientation} to also output \texttt{false}.
	
	We now analyze the running time. Constructing $(G^o, u^o)$ takes $O(|E(G)|)$ time because we subdivide each edge at most once. \textsc{FindEFXOrientObj} on Line~\ref{line:3.subroutine} takes $O(|V(G^o)|^2)$ time by Lemma~\ref{lemma:reduc:FindEFXOrientObj}(3). Since $G^o$ is constructed by subdividing each edge of $G$ at most once, and each subdivision introduces a new vertex, we have $|V(G^o)| \leq |V(G)| + |E(G)|$. Hence, $O(|V(G^o)|^2) = O((|V(G)|+|E(G)|)^2)$. Finally, constructing $\pi$ using $\pi^o$ takes $O(|E(G)|)$ time because it requires orienting each edge of $G$. Thus, \textsc{FindEFXOrientation} runs in $O((|V(G)|+|E(G)|)^2)$ time.
\end{proof}

\section{EF1 and EFX$_0$ Orientations of Multigraphs}\label{section:multigraphs}

We turn our attention to multigraphs and consider the problem of deciding if a multigraph has EF1 or EFX$_0$ orientations. We show that both of these problems are NP-complete using a reduction from the NP-complete problem \textsc{Partition} \citep{karp1972reducibility}.

We present two different reductions. The first reduction (Theorem~\ref{thm:main-multi}) is relatively simple but results in a multigraph with two self-loops. The second reduction (Theorem~\ref{thm:main-no-self-loop}) is more complex but has the advantage of not using self-loops.

\begin{theorem}\label{thm:main-multi}
	Deciding whether a multigraph $G$ of chores has an EF1 (resp.\ EFX$_0$) orientation is NP-complete, even if $G$ is symmetric, has only two vertices, and utility functions are additive.
\end{theorem}
\begin{proof}
	Both of these problems are in NP because verifying whether a given orientation of a multigraph of chores is EF1 or EFX$_0$ can clearly be done in polynomial time.
	
	We reduce an instance $S = \{s_1, s_2, \dots, s_k\}$ of \textsc{Partition} to each of the EF1 and EFX$_0$ orientation problems. We give the EF1 reduction and show how to adapt it for EFX$_0$.
	
	Let $\alpha < -\max_{i} s_i < 0$. Construct a multigraph $G$ on two vertices $a, b$ as follows. For each $s_i \in S$, create an edge $e_i$ between $a$ and $b$ and set $u_a(e_i) = u_b(e_i) = -s_i$. Create two self-loops $e_a, e_b$ at each of $a, b$, respectively, and set $u_a(e_a) = u_b(e_b) = \alpha$. Assume the utility functions are additive.
	
	Partitions $S_a \cup S_b$ of $S$ correspond to orientations of the edges between $a$ and $b$ in a natural way. Moreover, because both $a$ and $b$ receive their respective self-loop which has $\alpha < -\max_{i} s_i < 0$ utility, the EF1 criteria requires that their envy for each other be alleviated when they ignore their respective self-loops. Thus,  $\sum_{x \in S_a} x = \sum_{x \in S_b} x$ if and only if $G$ has an EF1 orientation.
	
	To adapt the reduction for EFX$_0$, let $\alpha = 0$ instead.
\end{proof}

We now show that deciding whether EF1 orientations exist remain NP-complete even if the multigraph $G$ contains no self-loops. To do this, we again rely on a reduction from \textsc{Partition}. For any instance $S = \{s_1, s_2, \dots, s_k\}$ of \textsc{Partition}, we define a multigraph $G$ on three vertices $a, b, c$ as follows. All of the edges $G$ contains has equal utility to both endpoints (called its {\em weight}). For each $s_i \in S$, create an edge between $a$ and $b$ of weight $-s_i$. Then, between $c$ and each of $a, b$, create two edges of weight $-T \coloneqq \sum_{i \in [k]} s_i$. (See Figure~\ref{fig:multi-no-loop}.)

\begin{lemma}\label{lemma:multigraph-reduction}
	In any EF1 orientation of $G$, at least one edge between $c$ and $a$ is oriented toward $a$ (and similarly between $c$ and $b$ by symmetry).
\end{lemma}
\begin{proof}
	Let $\pi = (\pi_a, \pi_b, \pi_c)$ be an EF1 orientation of $G$. If both edges between $c$ and $a$ are oriented toward $c$, then for either one of them (call it $e$), we have $u_c(\pi_c \setminus \{e\}) \leq -T < 0 = u_c(\pi_a)$, so $\pi$ is not EF1.
\end{proof}

\begin{theorem}\label{thm:main-no-self-loop}
	Deciding whether a multigraph $G$ of chores has an EF1 orientation is NP-complete, even if $G$ is symmetric, has only three vertices and no self-loops, and utility functions are additive.
\end{theorem}

\begin{proof}
	We have already seen in the proof of Theorem~\ref{thm:main-no-self-loop} that this problem is in NP. We now show that $S$ has an equipartition if and only if $G$ has an EF1 orientation. Suppose $S = S_1 \cup S_2$ is an equipartition. Let $\pi$ be the orientation that orients the edges between $a$ and $b$ in the way that naturally corresponds to $S_1 \cup S_2$, the two edges between $c$ and $a$ in opposite directions (one toward $c$ and the other toward $a$), and the two edges between $c$ and $b$ in opposite directions. We claim $\pi$ is EF1.
	
	Consider the envy that $c$ experiences. It suffices to consider the envy from $c$ to $a$ by symmetry. By construction, $u_c(\pi_c) = -2T < -T = u_c(\pi_a)$. Because both edges directed toward $c$ have utility $-T$, disregarding either of them alleviates the envy from $c$ to $a$.
	
	Next, consider the envy that $a$ and $b$ experience. Again by symmetry, it suffices to consider only $a$. The edge $e$ between $a$ and $c$ directed toward $a$ has the greatest negative utility among the edges in $\pi_a$ by the definition of $T$, so the definition of EF1 requires the envy $a$ experiences to be alleviated if $e$ is disregarded. Indeed, we have $u_a(\pi_a \setminus \{e\}) = (-1/2)T = u_a(\pi_b)$ and $u_a(\pi_a \setminus \{e\}) = (-1/2)T < -T = u_a(\pi_c)$. So, $\pi$ is EF1.
	
	Conversely, suppose $S$ has no equipartition. Let $\pi$ be any orientation of $G$ and $S_1 \cup S_2$ be the partition of $S$ that corresponds to the way $\pi$ orients the edges between $a$ and $b$. More specifically, let $s_i \in S_1$ if and only if the edge between $a$ and $b$ corresponding to $s_i$ is directed toward $a$, and let $S_2 = S \setminus S_1$.
	
	If both edges between $c$ and $a$ are directed toward $c$ or if both edges between $c$ and $b$ are directed toward $b$, then $\pi$ is not EF1 by Lemma~\ref{lemma:multigraph-reduction}, so we are done. It remains to consider the case that an edge $e_{ac}$ between $c$ and $a$ is directed toward $a$ and an edge $e_{bc}$ between $c$ and $b$ is directed toward $b$. The definition of EF1 requires $-\sum_{x \in S_1} x = u_a(\pi_a \setminus \{e_{ac}\}) \leq u_a(\pi_b) = -\sum_{x \in S_2} x$. Similarly, $-\sum_{x \in S_2} x = u_b(\pi_b \setminus \{e_{bc}\}) \leq u_b(\pi_a) = -\sum_{x \in S_1} x$. So, $\sum_{x \in S_1} = \sum_{x \in S_1} x$, a contradiction as $S$ has no equipartition.
\end{proof}

\begin{figure}
	\centering
	\begin{tikzpicture}
		\node[blkvertex] (a) at (-2, 0) {};
		\node[blkvertex] (b) at (2, 0) {};
		\node[blkvertex] (c) at (0, 2.5) {};
		
		\node (la) at (-2.3, 0) {$a$};
		\node (lb) at (2.3, 0) {$b$};
		\node (lc) at (0, 2.8) {$c$};

		\draw (a) to[bend left=35] node [left] {$-T$} (c);
		\draw (a) to[bend right=15] node [left] {$-T$} (c);
		
		\draw (b) to[bend left=15] node [right] {$-T$} (c);
		\draw (b) to[bend right=35] node [right] {$-T$} (c);
		
		\draw (a) to[bend left=25] node [above] {$-s_1$} (b);
		\draw (a) to node [above] {$-s_2$} (b);
		\node at (0, -0.2) {$\vdots$};
		\draw (a) to[bend right=25] node [below] {$-s_k$} (b);
	\end{tikzpicture}
	\caption{The multigraph $G$ in Theorem~\ref{thm:main-no-self-loop}. Here, $T \coloneqq \sum_{i \in [k]} s_i$.}
	\label{fig:multi-no-loop}
\end{figure}

\section{On the Additive Case}\label{section:additive}

All of our results also apply to the additive case in which every agent has an additive utility function, without any modification to the theorem statements and proofs. More generally, our results hold even if we allow each agent to have either a monotone utility function as previously described in Section~\ref{section:prelim} or an additive utility function, independently of other agents. We explain the reason below.

Our results pertaining to EF1 orientations stem from Proposition~\ref{prop:EF1-graphs}, which translates the EF1 orientation condition for each agent into a logically equivalent graph theoretical condition. Thus, any role that assumptions made on utility functions play lie entire within the proof of Proposition~\ref{prop:EF1-graphs}. It is easy to verify the correctness of this proof for additive utility functions as well.

As for EFX$_0$ orientations, Proposition~\ref{prop:one-neg-or-all-dummies} translates the EFX$_0$ orientation condition for each agent in an objective instances into an equivalent graph theoretical condition. So, we need to verify Proposition~\ref{prop:one-neg-or-all-dummies} and also the reduction of \textsc{EFX$_0$-Orientation} to \textsc{EFX$_0$-Orientation-Objective} for additive utility functions (i.~e.\ Lemma~\ref{lemma:reduc:FindEFXOrientation}). Again, it is straightforward to verify the proofs of both of these results for additive utility functions.

\section{Conclusion}\label{section:conclusion}

In this paper, we determined the complexity of finding EF1 and EFX$_0$ orientations of chores in both the graph and multigraph settings, and gave polynomial-time algorithms for the polynomial-time cases.


\begin{ack}
	We acknowledge the support of the Natural Sciences and Engineering Research Council of Canada (NSERC), funding reference number RGPIN-2022-04518. We also thank the anonymous reviewers for their constructive comments. 
\end{ack}


\bibliography{bibliography}

\begin{thebibliography}{18}
\providecommand{\natexlab}[1]{#1}
\providecommand{\url}[1]{\texttt{#1}}
\expandafter\ifx\csname urlstyle\endcsname\relax
  \providecommand{\doi}[1]{doi: #1}\else
  \providecommand{\doi}{doi: \begingroup \urlstyle{rm}\Url}\fi

\bibitem[Afshinmehr et~al.(2024)Afshinmehr, Danaei, Kazemi, Mehlhorn, and
  Rathi]{afshinmehr2024efx}
M.~Afshinmehr, A.~Danaei, M.~Kazemi, K.~Mehlhorn, and N.~Rathi.
\newblock {EFX} allocations and orientations on bipartite multi-graphs: A
  complete picture.
\newblock \emph{arXiv preprint arXiv:2410.17002}, 2024.

\bibitem[Akrami et~al.(2024)Akrami, Alon, Chaudhury, Garg, Mehlhorn, and
  Mehta]{akrami2024efx}
H.~Akrami, N.~Alon, B.~R. Chaudhury, J.~Garg, K.~Mehlhorn, and R.~Mehta.
\newblock {EFX}: a simpler approach and an (almost) optimal guarantee via
  rainbow cycle number.
\newblock \emph{Operations Research}, 2024.

\bibitem[Aspvall et~al.(1979)Aspvall, Plass, and Tarjan]{aspvall1979linear}
B.~Aspvall, M.~F. Plass, and R.~E. Tarjan.
\newblock A linear-time algorithm for testing the truth of certain quantified
  boolean formulas.
\newblock \emph{Information processing letters}, 8\penalty0 (3):\penalty0
  121--123, 1979.

\bibitem[Bhaskar and Pandit(2024)]{bhaskar2024efx}
U.~Bhaskar and Y.~Pandit.
\newblock {EFX} allocations on some multi-graph classes.
\newblock \emph{arXiv preprint arXiv:2412.06513}, 2024.

\bibitem[Caragiannis et~al.(2019)Caragiannis, Kurokawa, Moulin, Procaccia,
  Shah, and Wang]{caragiannis2019unreasonable}
I.~Caragiannis, D.~Kurokawa, H.~Moulin, A.~D. Procaccia, N.~Shah, and J.~Wang.
\newblock The unreasonable fairness of maximum {N}ash welfare.
\newblock \emph{ACM Transactions on Economics and Computation (TEAC)},
  7\penalty0 (3):\penalty0 1--32, 2019.

\bibitem[Chaudhury et~al.(2020)Chaudhury, Garg, and Mehlhorn]{chaudhury2020efx}
B.~R. Chaudhury, J.~Garg, and K.~Mehlhorn.
\newblock {EFX} exists for three agents.
\newblock In \emph{Proceedings of the 21st ACM Conference on Economics and
  Computation}, pages 1--19, 2020.

\bibitem[Christodoulou et~al.(2023)Christodoulou, Fiat, Koutsoupias, and
  Sgouritsa]{christodoulou2023fair}
G.~Christodoulou, A.~Fiat, E.~Koutsoupias, and A.~Sgouritsa.
\newblock Fair allocation in graphs.
\newblock In \emph{Proceedings of the 24th ACM Conference on Economics and
  Computation}, pages 473--488, 2023.

\bibitem[Deligkas et~al.(2024)Deligkas, Eiben, Goldsmith, and
  Korchemna]{deligkas2024ef1}
A.~Deligkas, E.~Eiben, T.-L. Goldsmith, and V.~Korchemna.
\newblock {EF1} and {EFX} orientations.
\newblock \emph{arXiv preprint arXiv:2409.13616}, 2024.

\bibitem[Gourv{\`e}s et~al.(2014)Gourv{\`e}s, Monnot, and
  Tlilane]{gourves2014near}
L.~Gourv{\`e}s, J.~Monnot, and L.~Tlilane.
\newblock Near fairness in matroids.
\newblock In \emph{ECAI}, volume~14, pages 393--398, 2014.

\bibitem[Hosseini et~al.(2023)Hosseini, Sikdar, Vaish, and
  Xia]{hosseini2023fairly}
H.~Hosseini, S.~Sikdar, R.~Vaish, and L.~Xia.
\newblock Fairly dividing mixtures of goods and chores under lexicographic
  preferences.
\newblock In \emph{Proceedings of the 2023 International Conference on
  Autonomous Agents and Multiagent Systems}, pages 152--160, 2023.

\bibitem[Hsu()]{hsu2024efx}
K.~Hsu.
\newblock {EFX} orientations of multigraphs.
\newblock To appear in ECAI 2025.

\bibitem[Karp(1972)]{karp1972reducibility}
R.~Karp.
\newblock Reducibility among combinatorial problems.
\newblock \emph{Complexity of Computer Computations}, pages 85--103, 1972.

\bibitem[Li et~al.(2024)Li, Li, Wei, Wu, and Zhou]{li2024complete}
B.~Li, M.~Li, T.~Wei, Z.~Wu, and Y.~Zhou.
\newblock A complete landscape of {EFX} allocations on graphs: Goods, chores
  and mixed manna.
\newblock \emph{arXiv preprint arXiv:2409.03594}, 2024.

\bibitem[Plaut and Roughgarden(2020)]{plaut2020almost}
B.~Plaut and T.~Roughgarden.
\newblock Almost envy-freeness with general valuations.
\newblock \emph{SIAM Journal on Discrete Mathematics}, 34\penalty0
  (2):\penalty0 1039--1068, 2020.

\bibitem[Procaccia(2020)]{procaccia2020technical}
A.~D. Procaccia.
\newblock Technical perspective: An answer to fair division's most enigmatic
  question.
\newblock \emph{Communications of the ACM}, 63\penalty0 (4):\penalty0 118--118,
  2020.

\bibitem[Sgouritsa and Sotiriou(2025)]{sgouritsa2025existence}
A.~Sgouritsa and M.~M. Sotiriou.
\newblock On the existence of {EFX} allocations in multigraphs.
\newblock \emph{arXiv preprint arXiv:2502.09777}, 2025.

\bibitem[Zeng and Mehta(2024)]{zeng2024structure}
J.~A. Zeng and R.~Mehta.
\newblock On the structure of envy-free orientations on graphs.
\newblock \emph{arXiv preprint arXiv:2404.13527}, 2024.

\bibitem[Zhou et~al.(2024)Zhou, Wei, Li, and Li]{ijcai2024p338}
Y.~Zhou, T.~Wei, M.~Li, and B.~Li.
\newblock A complete landscape of {EFX} allocations on graphs: Goods, chores
  and mixed manna.
\newblock In K.~Larson, editor, \emph{Proceedings of the Thirty-Third
  International Joint Conference on Artificial Intelligence, {IJCAI-24}}, pages
  3049--3056. International Joint Conferences on Artificial Intelligence
  Organization, 8 2024.
\newblock \doi{10.24963/ijcai.2024/338}.
\newblock URL \url{https://doi.org/10.24963/ijcai.2024/338}.
\newblock Main Track.

\end{thebibliography}

\end{document}